\documentclass[11pt]{jmpclass}
\newboolean{soda}
\setboolean{soda}{false}
 
\usepackage{amssymb}
\usepackage{graphicx}
\usepackage{citesort}
\usepackage{multirow}
\usepackage{ifthen}
 \usepackage{graphicx} 
\usepackage{wrapfig}
\usepackage{ifpdf}
\ifpdf\fi

\renewcommand{\paragraph}[1]{\medskip\noindent{\bf #1}}

\newtheorem{invariant}{Invariant}
 
\allowdisplaybreaks
\newcommand{\Bcal}{\mathcal{B}}
\renewcommand{\abs}[1]{\left| #1 \right|}

\title{Equivalence between Priority Queues and Sorting in External Memory
  \\ 
}
\author{Zhewei Wei \\ HKUST \and Ke Yi \\ HKUST}

\begin{document}

\begin{titlepage}
\maketitle

\begin{abstract}

A priority queue is a fundamental data structure that maintains a dynamic
ordered set of keys and supports the followig basic operations: insertion
of a key, deletion of a key, and finding the smallest key.  The complexity
of the priority queue is closely related to that of sorting: A priority
queue can be used to implement a sorting algorithm trivially.
Thorup~\cite{thorup2007equivalence} proved that the converse is also true
in the RAM model. In particular, he designed a priority queue that uses the
sorting algorithm as a black box, such that the per-operation cost of the
priority queue is asymptotically the same as the per-key cost of sorting.
In this paper, we prove an analogous result in the external memory model,
showing that priority queues are computationally equivalent to sorting in
external memory, under some mild assumptions.  The reduction provides a
possibility for proving lower bounds for external sorting via showing a
lower bound for priority queues.


\end{abstract}
\end{titlepage}


\section{Introduction}
\label{sec:introduction}
The priority queue is an abstract data structure of fundamental importance.
A priority queue maintains a set of keys and support the following
operations: insertion of a key, deletion of a key, and findmin, which
returns the current minimum key in the priority queue.  It is well known
that a priority queue can be used to implement a sorting algorithm: we
simply insert all keys to be sorted into the priority queue, and then
repeatedly delete the minimum key to extract the keys in sorted order.
Thorup \cite{thorup2007equivalence} showed that the converse is also true
in the RAM model. In particular, he showed that given a sorting algorithm
that sorts $N$ keys in $NS(N)$ time, there is a priority queue that uses
the sorting algorithm as a black box, and supports insertion and deletion
in $O(S(N))$ time, and findmin in constant time.  The reduction uses linear
space.  The main implication of this reduction is that we can regard the
complexity of internal priority queues as settled, and just focus on
establishing the complexity of sorting.  Algorithmically, it also gives new
priority queue constructions by using the fastest (integer) sorting
algorithms currently known: an $O(N\log\log N)$ deterministic algorithm by
Han~\cite{han2004deterministic} and an $O(N\sqrt{\log\log N})$ randomized
one by Han and Thorup~\cite{han2002integer}.

In this paper, we prove an analogous result in the external memory model
(the I/O model), showing that priority queues are almost computationally
equivalent to sorting in external memory. We design a priority queue that
uses the sorting algorithm as a black box, such that the update cost of the
priority queue is essentially the same as the per-key I/O cost of the
sorting algorithm.  The priority queue always has the current minimum key
in memory so findmin can be handled without I/O cost.  Our priority queue
is a non-trivial generalization of Thorup's, which is fundamentally an
internal structure.  The main reasons why Thorup's structure does not work
in the I/O model are that it cannot flush the buffers I/O-efficiently, and
that it does not specify any order for performing the flush and rebalance
operations. Moreover, deletions are supported in a very different way in
the I/O model; we have to do it in a lazy fashion in order to achieve
I/O-efficiency.

\subsection{Our results} 
Let us first recall the standard I/O model \cite{aggarwal:input}: The
machine consists of an internal memory of size $M$ and an infinitely large
external memory.  Computation can only be carried out in internal memory.
The external memory is divided into blocks of size $B$, and data is moved
between internal and external memory in terms of blocks.  We measure the
complexity of an algorithm by counting the number of I/Os it performs,
while internal memory computation is free.

Our main result is stated in the following theorem:
\begin{theorem}  
\label{thm:priority_queue}
Suppose we can sort up to $N$ keys in $NS(N)/B$ I/Os in external memory,
where $S$ is a non-decreasing function. Then there exists an external
priority queue that uses linear space and supports a sequence of $N$
insertion and deletion operations in $O(\frac{1}{B}\sum_{i\ge
  0}S(B\log^{(i)}\frac{N}{B}))$ amortized I/Os per operation.  Findmin can
be supported without I/O cost.  The reduction uses $O(B)$ internal memory
and is deterministic.
\end{theorem}

The first implication of Theorem~\ref{thm:priority_queue} is that if the
main memory has size $\Omega(B\log^{(c)}\frac{N}{B})$ for any constant $c$,
then our priority queue supports insertion and deletion with $O(S(N)/B)$ amortized
I/O cost. This is because $S(N)=0$ when $N\le M$.  Even if $M=O(B)$, the
reduction is still tight as long as the function $S$ grows not too slowly.
More precisely, we have the following corollary:
\begin{corollary}
\label{cor:priority_queue}
For $S(N)= \Omega(2^{ \log^{*}\frac{N}{B} })$, the priority queue supports
updates with $O(S(N)/B)$ amortized I/O cost; for $S(N)=o(2^{
  \log^{*}\frac{N}{B} })$, the priority queue supports updates with
$O(S(N)\log^*\frac{N}{B}/B)$ amortized I/O cost.
\end{corollary}
The first part can be verified by plugging $S(N)=2^{\log^*\frac{N}{B}}$
into Theorem~\ref{thm:priority_queue} and showing that the
$S(B\log^{(i)}\frac{N}{B})$'s decrease exponentially with $i$. For the
second part, we simply relax all the $S(B\log^{(i)}\frac{N}{B})$'s to
$S(N)$. Note that $2^{\log^* \frac{N}{B}}=o(\log^{(c)}\frac{N}{B})$ for any
constant $c$, so it is very unlikely that a sorting algorithm could achieve
$S(N)=o(2^{ \log^{*}\frac{N}{B} })$.  No such algorithm is known, even in
the RAM model.  Therefore, we can essentially consider our reduction to be tight.

\subsection{Related work}
\label{subsec:related_work}
Sorting and priority queues have been well studied in the comparison-based
I/O model, in which the keys can only be accessed via comparisons. Aggarwal
and Vitter~\cite{aggarwal:input} showed that
$\Theta(\frac{N}{B}\log_{M/B}\frac{N}{B})$ I/Os are sufficient and
necessary to sort $N$ keys in the comparison-based I/O model.  This bound
is often referred to as the {\em sorting bound}. If the comparison
constraint is replaced by the weaker indivisibility constraint, there is an
$\Omega(\min\{\frac{N}{B}\log_{M/B}\frac{N}{B},N\})$ lower bound, known as
the {\em permuting bound}. The two bounds are the same when
$\frac{N}{B}\log_{M/B}\frac{N}{B}<N$; it is conjectured that for this
parameter range, $\Omega(\frac{N}{B}\log_{M/B}\frac{N}{B})$ is still the
sorting lower bound even without the indivisibility constraint. For
$\frac{N}{B}\log_{M/B}\frac{N}{B}>N$, the current situation in the I/O
model is the same as that in the RAM model, that is, the best upper bound
is just to use the best RAM algorithm (which has $O(N\log\log N)$ time
deterministically or $O(N\sqrt{\log\log N})$ time randomized) naively in
external memory ignoring the blocking at all, and there is no non-trivial
lower bound.  When the block size is not too small, none of the RAM sorting
algorithms works better than the comparison-based one, which makes the
situation ``cleaner''.  Thus, a sorting lower bound (without any
restrictions) has been considered to be more hopeful in the I/O model (with
$B$ not too small) than in the RAM model, and it was posed as a major open
problem in~\cite{aggarwal:input}. Thus, our result provides a way to
approach a sorting lower bound via that of priority queues, while data
structure lower bounds have been considered (relatively) easier to obtain
than (concrete) algorithm lower bounds (except in restricted computation
models), as witnessed by the many recent strong cell probe lower bounds for
data structures, such as \cite{patrascu:unifying,larsen12:_cell} among many
others.  However, our result does not offer any new bounds for priority
queues because we do not know of a better sorting algorithm than the
comparison-based ones in the I/O model (and the conjecture is that they do
not exist when $\frac{N}{B}\log_{M/B}\frac{N}{B}<N$).

Since a priority queue can be used to sort $N$ keys with $N$ insertion and
$N$ deletemin operations, it follows that
$\Omega(\frac{1}{B}\log_{M/B}\frac{N}{B})$ is also a lower bound for the
amortized I/O cost per operation for any external priority queue, in the
comparison-based I/O model.  There are many priority queue constructions
that achieve this lower bound, such as the buffer tree~\cite{arge:buffer},
$M/B$-ary heaps~\cite{fadel1999heaps}, and array
heaps~\cite{brodal1998worst}.  See the survey~\cite{vitter2001external} for
more details.  However, they do not use sorting as just a black box, and
cannot be improved even if we have a faster external sorting algorithm.
Thus they do not give a priority queue-to-sorting reduction.  The extra
$O(\log_{M/B}\frac{N}{B})$ factor comes from a tree structure with fanout
$O(M/B)$ within the priority queue construction. and a key must be moved
$\Omega(\log_{M/B}\frac{N}{B})$ times to ``bubble up" or ``bubble down".

Arge et al.~\cite{arge2002cache} developed a cache-oblivious priority
queue that achieves the sorting bound with the tall cache assumption, that
is, $M$ is assumed to be of size at least $B^2$.  We note that their
structure can serve as a priority queue-to-sorting reduction in the I/O
model, by replacing the cache-oblivious sort with a sorting black box. The
resulting priority queue supports all operations in
$O(\frac{1}{B}\sum_{i\ge 0}S(N^{(2/3)^i })$ amortized I/Os if the
sorting algorithm sorts $N$ keys in $NS(N)/B$ I/Os.  However, this
reduction is not tight for $S(N)=O(\log\log\frac{N}{B})$, and there seems
to be no easy way to get rid of the tall cache assumption, even if the
algorithm has the knowledge of $M$ and $B$.

\section{Structure} 
In this section, we describe the structure of our priority queue.  In the
next section, we show how this structure supports various operations.
Finally we analyze the I/O costs of these operations.

The priority queue consists of multiple layers whose sizes vary from $N$ to
$cB$, where $c$ is some constant to be determined later.  The $i$'th layer
from above has size $\Theta(B\log^{(i)}\frac{N}{B})$, for $i\ge 0$, and the
priority queue has $O(\log^*N)$ layers. For the sake of simplicity we will
refer to a layer by its size.  Thus the layers from the largest to the
smallest are layer $N$, layer $B\log\frac{N}{B}$, $\ldots$, layer $cB$.
Layer $cB$ is also called {\em head}, and is stored in main memory. Given a
layer $X$, its {\em upper layer} and {\em lower layer} are layer
$B2^{\frac{X}{B}}$ and  layer $B\log\frac{X}{B}$, respectively.  We
use $\Psi_X$ to denote $B2^{\frac{X}{B}}$ and $\Phi_X$ to denote
$B\log\frac{X}{B}$.  The priority queue maintains the invariant that the
keys in layer $\Phi_X$ are smaller than the keys in layer $X$.  In
particular, the minimum key is always stored in the head and can be
accessed without I/O cost.

\begin{figure}[b!]
\centering 
\includegraphics[width=\textwidth]{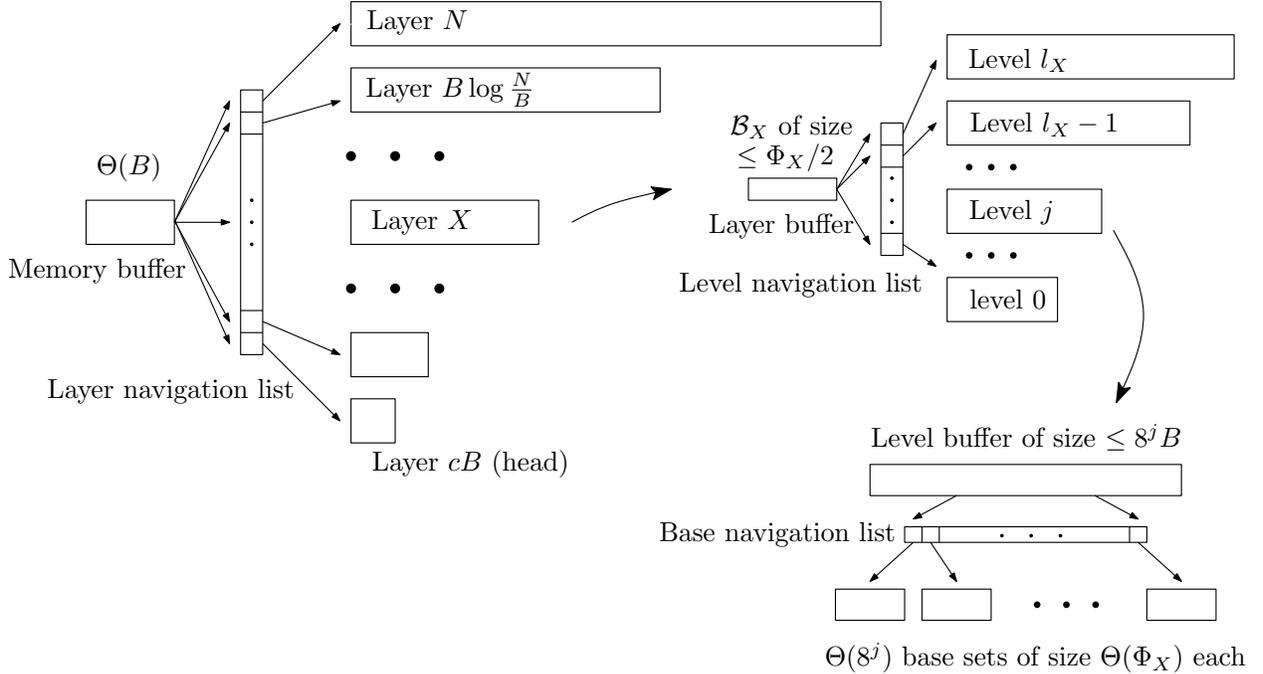}
\caption{\label{fig:layers} The components of the priority queue.}
\end{figure} 

We maintain a main memory buffer of size $O(B)$ to accommodate incoming
insertion and deletion operations. In order to  distribute keys in the memory buffer to
different layers I/O-efficiently, we maintain a structure called
 {\em layer navigation list}. Since this structure
will also be used in other components of the priority queue, we define
it in a unified way. Suppose we want to
distribute the keys in a buffer $\Bcal$ to $t$
sub-structures $S_1, S_2,\ldots S_t$. The keys in different
sub-structures are sorted relative to each other, that is, the keys in
$S_i$ are less or equal to the keys in $S_{i+1}$. Each sub-structure $S_i$
is associated with a buffer $\Bcal_i$, which accommodates
 keys transferred from $\Bcal$. 
The goal is to distribute the keys in $\Bcal$ to each   
$\Bcal_i$ I/O-efficiently, such that the keys that 
go to $\Bcal_i$ have values between the minimum keys of $S_i$ and
$S_{i+1}$. A navigation list stores a set
of $t$ {\em representatives}, each representing a sub-structure. The
representative of $S_i$, denoted $r_i$, is a triple that stores 
the minimum key of $S_i$, the number of keys stored in $\Bcal_i$, 
and a pointer to the last
non-full block of the buffer $\Bcal_i$. The
representatives are stored consecutively on the disk, and are sorted on the
minimum keys. The layer navigation list is built for the $O(\log^*
N)$ layers, so it has size $O(\log^* N)$. Please see Figure~\ref{fig:layers}.
   
Now we will describe the structures inside  a layer $X$ except layer
$cB$, which is always in the main memory. 
 First we maintain a {\em layer buffer} of size $\Phi_X/2$ to store 
keys flushed from the memory buffer. 
The main structure of layer
$X$ consists of  $O(\log \frac{X}{\Phi_X})$ levels with
exponentially increasing sizes. The $j$'th level from the bottom,
denoted level $j$, has size  $\Theta(8^j\Phi_X)$.
We also keep the invariant that the keys in level $j$ are
less or equal to the keys in level $j+1$.
We maintain a {\em level navigation list} of size $\Theta(\log\frac{X}{\Phi_X})$, 
which represents the $\log\frac{X}{\Phi_X}$ levels.
Most keys in level $j$ are stored in  $\Theta(8^j)$
disjoint {\em base sets}, each of size $\Theta(\Phi_X)$. 
The base sets, from left to right, are sorted relative to each other, but
they are not internally sorted. Other than the base sets, there is a {\em level
  buffer} of size $8^jB$, which is used to temporarily accommodate keys before
distributing them to the base set.  We also maintain a {\em base navigation
  list} of  size $\Theta(8^j)$ for the  base sets. Note
that we do not impose the level structures on layer $cB$ since it can fit in
the main memory.
The components of the priority queue are illustrated in Figure~\ref{fig:layers}.

Here we provide some intuition  for this  complicated
structure. We first divide the keys into exponentially increasing levels,
which in some sense is similar to building a heap. 
However, the $O(\log N)$-level structure implies
that a key may be moved $O(\log N)$ times in its lifetime. To
overcome this,  we
group the keys into base sets of logarithmic sizes so that we can
move more keys with the same I/O cost (we will move pointer to the
base sets rather than the base sets themselves).  
Finally, when a base set gets down to level $0$,  we need to recursively build the structure on it, which
results in the $O(\log^* N)$ layers. 
     

Let $l_X$ denote the top level of layer $X$. 
We use $\Bcal_X$ to denote the layer buffer of layer $X$
and $\Bcal_j$ to denote the level buffer of level $j$ when the layer
is specified. Our priority queue
maintains the following invariants for layer $X$: 
 
\begin{invariant}
\label{inv:buffer_size}
The layer buffer $\Bcal_{X}$ contains at most
$\frac{1}{2}\Phi_X$ keys; the level buffer $\Bcal_j$ at 
layer $X$ contains at most $8^jB$ keys. 
\end{invariant}
\begin{invariant}
\label{inv:buffer_order}
The layer buffer $\Bcal_X$ only contains keys between
the minimum keys of layer $X$ and  its upper layer.
The level buffer $\Bcal_j$ only contains keys
between the minimum keys of level $j$ and its upper level.  
\end{invariant}

\begin{invariant} 
\label{inv:level_size}
A  base set in layer $X$ has size between $\frac{1}{2}
\Phi_X$ and $2\Phi_X$;
level $j$ of layer $X$, for $j=0, 1, \ldots, l_X-1$, has size between $2\cdot 8^j \Phi_X$
and $6 \cdot 8^j \Phi_X$, and level $l_X$ has size  between $2\cdot
8^{l_X}\Phi_X$ and $40\cdot 8^{l_X}\Phi_X$.  
\end{invariant}

\begin{invariant}
\label{inv:head}
The head contains at most $2cB$ keys.
\end{invariant}

Note that when we talk about the size of a level, we only 
count the keys in its base sets and exclude the level buffer. 
The top level has a slightly different size range so
that  the construction works for any value of  $X$. 

We say a layer buffer, a level buffer, a base set, a level or the head
overflows if its size exceeds its upper bound in
Invariant~\ref{inv:buffer_size}, ~\ref{inv:level_size}
or~\ref{inv:head}; we say a base set, a level underflows
if its size gets below the lower bound in
Invariant~\ref{inv:level_size}.


\section{Operations}
Recall that the priority queue supports three operations: insertion,
deletion, and findmin. Since we always maintain the minimum key in the main
memory (it is in either the head or the memory buffer), the cost of  
a findmin operation is free. We process deletions in a lazy fashion,
that is, when a deletion comes we generate a {\em delete signal} with the
corresponding key and a time stamp, and insert the delete signal to 
the priority queue. In most cases we treat the delete signals as norm insertions. We
only perform the actually delete in the head so that the current 
minimum key is always valid. To ensure 
linear space usage we perform a global rebuild after every $N/8$ updates.
  
Our priority queue is implemented by three general operations:
{\em global rebuild}, {\em flush}, and {\em rebalance}. A global
rebuild operation sorts all keys and processes all delete signals to
maintain linear size. A flush operation
distributes all keys in a buffer to the buffers of corresponding
sub-structures to maintain Invariant~\ref{inv:buffer_size}. A rebalance operation 
moves keys between two adjacent 
sub-structures  to maintain Invariant~\ref{inv:level_size}.

\subsection{Global Rebuild}
We conduct the first global rebuild when the internal memory buffer is
full. Then, after each global rebuild, we set $N$ to be the number of
keys in the priority queue, and keep it fixed until the next global
rebuild. 
A global rebuild is triggered whenever layer $N$ (in fact, its top
level) becomes unbalanced or
the priority queue has received $N/8$ new updates since the last global rebuild.
We  show that it
takes  $O(NS(N)/B)$ I/Os to rebuild our priority queue.  
We first sort all keys in the priority queue and process the delete
signals. Then we scan through the remaining keys and divide them into base
sets of size $\Phi_N$, except the last base set
which may be smaller. This base set is merged to its predecessor  if
its size is less or equal to $\frac{1}{2}\Phi_N$.  The first
base set is used to construct the lower layers, and the rest are used to construct layer
$N$. To rebuild the $O(\log\frac{N}{\Phi_N})$ levels of layer $N$, 
we scan through the base sets, and
take the next $4\cdot 8^j$ base sets to build level $j$, for
$j=0,1,2,\ldots$. 
Note that the base
navigation list of these $4\cdot 8^j$ base sets can be constructed when
we scan through the keys in the base sets. The level rebuild process stops
when  we encounter an integer
$l_N$ such that the number of remaining base sets  is more than
$4\cdot 8^{l_N}$, but less or equal to $4\cdot( 8^{l_N}+8^{l_N+1})=36\cdot 8^{l_N}$.
Then we take these base sets to form the top level of layer $N$. 
After the global rebuild, level $j$ has size $4^j\Phi_N$, and the top
level $l_N$ has size between $(4\cdot 8^{l_N}-\frac{1}{2})\Phi_N$ and $(36\cdot
8^{l_N}+\frac{1}{2})\Phi_N$. For $X= B\log\frac{N}{B},
B\log^{(2)}\frac{N}{B},\ldots,cB $,  layer $X$ are constructed 
recursively using the same algorithm. All buffers are left empty. 

Based
on the global rebuild algorithm, the priority queue maintains the
following invariant between two global rebuilds:
\begin{invariant}
\label{inv:number_of_levels}
The top level $l_X$ in layer $X$ is determined by the
maximum $l_X$ such that 
$$1+\sum_{j=0}^{l_X}4\cdot 8^j\le \frac{X}{\Phi_X}.$$
The number of layers and the number of levels in each layer will
not change between two global rebuilds. 
\end{invariant}
As a result of Invariant~\ref{inv:number_of_levels}, we have the
following lemma:
\begin{lemma}
\label{lem:top_level}
Suppose the top level in layer $X$ is level $l_X$. 
Then $l_X$ is an integer that
satisfies the following inequality:
$$4\cdot 8^{l_X}\Phi_X \le X \le 40\cdot 8^{l_X}\Phi_X.$$
\end{lemma}

\subsection{Flush}
We define the  flush operation in a unified way.
Suppose  we have a buffer $\Bcal$ and $k$
sub-structures $S_1, S_2,\ldots, S_k$. Each $S_i$ is  associated with a
buffer $\Bcal_i$, and a navigation list  $L$ of size $k$ is maintained
for the $k$ sub-structures. To flush the buffer $\Bcal$ we first sort
the  keys in it. Then we scan through the navigation list, 
and for each representative $r_i$ in $L$, we read the last non-full
block of $\Bcal_i$ to the memory, and fill it with keys in
$\Bcal$. When the block is full, we write it back to disk, and
allocate a new block. We do so
until we encounter a key
that is larger than the  key  in $r_{i+1}$. 
Then we update $r_i$, and advance to $r_{i+1}$. The I/O cost for a flush  
is the cost of sorting a buffer of
size $\abs{\Bcal}$ plus one I/O for each
sub-structure, so we have the following lemma:
\begin{lemma} 
\label{lem:flush}
The I/O cost for flushing keys in buffer $\Bcal$ to $k$ sub-structures
is bounded by
$O(\frac{\abs{\Bcal}S(\abs{\Bcal})}{B}+k)$.
\end{lemma}

There are three individual flush operations. A {\em memory flush}
distributes keys in the internal memory buffer to $O(\log^* N)$ layer buffers; a 
{\em layer  flush} on layer $X$ distributes keys in the layer buffer
to $O(\log\frac{X}{\Phi})$ level
buffers in the layer; and a {\em level flush} on level $j$ at layer $X$
distributes  keys in the level buffer to $\Theta(8^j)$ base sets in
the level.

\subsection{Rebalance} 
\paragraph{Rebalancing the base sets.}
Base rebalance is performed only after a level flush, since this is the
only operation that causes a base set to be unbalanced.  Consider a level
flush in level $j$ of layer $X$.  Suppose the base set $A$ overflows after
the flush. To rebalance $A$ we sort and scan through the keys in it, and
split it into base sets of size $\Phi_X$.  If the last base set has less
than $\frac{1}{2}\Phi_X$ keys we merge it into its predecessor.  Note that
any base set coming out of a split has between $\frac{3}{2}\Phi_X$ and
$\frac{1}{2}\Phi_X$ keys, so it takes at least $\frac{1}{2}\Phi_X$ new
updates to any of them before it initiates a new split. Note that after the
split we should update the representatives in the base navigation list.
This can be done without additional I/Os to the level flush operation: We
store all new representatives in a temporary list and rebuild the
navigation list after all overflowed base sets are rebalanced in level $j$.
A base set never underflows so we do not have a join operation.

\paragraph{Rebalancing the levels.} 
We define two level rebalance operations: {\em level push}
and {\em level pull}. Consider level $j$ at layer $X$. 
When the number of keys in level $j$ (except the top level)
gets to more than $6\cdot 8^j\Phi_X$,
a level push operation is performed to
 move some of its base sets to the upper level. More
precisely, we scan through the navigation list of level $j$ to find the first
representative  $r_k$ such that the number of keys before $r_k$ is
larger than $4\cdot 8^j\Phi_X$. Then we split the navigation list
of level $j$ around $r_k$ and attach the second half to the navigation list
of level $i+1$. Note that by moving the representatives we also move
their corresponding base sets to level $j+1$. By Invariant~\ref{inv:level_size}, 
the number of keys in a
base set is at most $2\Phi_X$, so the new level
$j$ has size  between $4\cdot 8^j\Phi_X$ and $(4\cdot
8^j+2)\Phi_X$.
Finally, to maintain Invariant~\ref{inv:buffer_order} we sort
level buffer $\Bcal_j$
and move keys larger than the $r_k$ to the level buffer
$\Bcal_{j+1}$. 

 Conversely, 
if the number of keys in level $j$ gets below $2\cdot 8^j\Phi_X$
(except the top level), a level
pull operation is performed. We cut
a proportion of the navigation list of level $j+1$ and attach it to the
navigation list of level
$j$, such that the number of keys in level $i$
becomes between $4\cdot 8^j\Phi_X$ and $(4\cdot
8^j-2)\Phi_X$. 
We also sort  $\Bcal_{j+1}$, the buffer
of level $j+1$, and move the corresponding keys  to level buffer $\Bcal_j$.

Observe that  after a
level push/pull, the number
of keys in  level $j$ is between $(4\cdot
8^j-2)\Phi_X$ and $(4\cdot 8^j+2)\Phi_X$, so it takes
at least $\Omega(8^j\Phi_X)$ new updates before the level
needs to be rebalanced again. 
The main reason that we adopt this level rebalance strategy is that
it does not touch all keys in the level; the rebalance only
takes place on the base navigation lists and the keys in the level buffers. 


\paragraph{Rebalancing the layers.}
When  the top level $l_X$ of  layer $X$ becomes
unbalanced, we can no longer rebalance it only using navigation list.
Recall that  its upper  level is level
$0$ in layer $\Psi_X$. For
simplicity we will refer to the the two levels as level $l_X$ and level
$0$, without specifying their layers. 
We also define two operations for rebalancing a layer: 
{\em layer push} and {\em layer  pull}. A layer push is performed when
the layer overflows, that is, 
the number of keys in level $l_X$ gets more than $40\cdot
8^{l_X}\Phi_X$. In this case   
we sort all keys in level $l_X$ and level $0$ together, then 
use the first $4\cdot 8^{l_X}\Phi_X$
keys to rebuild level $l_X$ and the rest to rebuild level $0$. Recall
that to rebuild a level we  scan through the keys and divide them into
base sets of size $\Phi_X$, except the last one which has size between
$\frac{1}{2}\Phi_X$ and $\frac{3}{2}\Phi_X$, and then 
we scan through the keys again to 
build the base navigation list. 
Note that the rebuild operation will change the minimum key
in layer $\Psi_X$, so we update the layer navigation list
accordingly. Finally we 
sort the keys in the layer buffer $\Bcal_X$ and the level buffer $\Bcal_{l_X}$,
and move the keys larger than the new minimum key of layer
$\Psi_X$ to the level buffer $\Bcal_0$. 

A layer pull operation is performed when the layer underflows, that
is, there are less than $2\cdot
8^{l_X}\Phi_X$ keys in level $l_X$. A layer pull proceeds in the same
way as a layer push does, except for the last step. Here we sort the layer
buffer $\Bcal_{\Psi_X}$ and the level buffer $\Bcal_{0}$ and
move the keys smaller than the new minimum key to the level buffer
$\Bcal_{l_X}$. After a layer push or pull, the number of keys in level
$l_X$ is $4\cdot 8^j\Phi_X$. 
By lemma~\ref{lem:top_level}, we have  $40\cdot 8^{l_X}\Phi_X \ge X$, so
it takes at least
$2\cdot 8^{l_X}\Phi_X=\Omega(X)$ new updates to layer $X$ before we
initiate a new push or a pull again.

Note that since we do not impose the level structure on the head layer $cB$, we
need to design the layer push and layer pull operations specifically for it. 
A layer push is performed when the number of keys in the head gets to
more than $2cB$. We sort
 all keys in it and level $0$ of layer $\Psi_{cB}$, and 
use the first $cB$ keys to rebuild the head and the
rest to rebuild level $0$.  A layer pull is performed when the head
becomes empty. The operation processes  in the
same way as a layer push does, except that after rebuilding both
levels, we sort the layer
buffer $\Bcal_{\Psi_{cB}}$ and the level buffer $\Bcal_0$ together,  and move the
keys smaller than the new minimum key of layer $\Bcal_{\Psi_{cB}}$ to
the head.


\subsection{Scheduling Flush and Rebalance Operations }
In order to achieve the I/O bounds in
Theorem~\ref{thm:priority_queue},
 we need to schedule the operations delicately. Whenever the memory
 buffer overflows we start to update the priority queue.
 This process is divided into three stages: the flush stage, 
the push stage, and the pull stage. In the flush stage we flush all
overflowed buffers and rebalance all unbalanced base sets; in the
push stage  we 
use push operations to rebalance all overflowed layers and levels. 
We treat delete signals as
insertions in the flush stage and the push stage.
In the pull stage we deal with delete signals and use pull
operations to rebalance all underflowed layers and levels.

In the flush stage, we initialize a queue $Q_o$ to keep track of all 
overflowed buffers and a
 doubly linked list $L_o$ to keep track of all overflowed levels. The buffers
 are flushed in a BFS fashion. First we flush the memory buffer
 into $O(\log^* N)$ layer buffers. After flushing the memory buffer,
we insert the representatives of  the overflowed layer buffers
into $Q_o$, from bottom to top.
 We also check whether the head overflows after the memory
flush. If so, we insert its representatives to the beginning of $L_o$.
Then we start to flush the layer buffers in $Q_o$.
 Again, when flushing a layer buffer we insert the representatives of the
 overflowed level buffers to $Q_o$ from bottom to top. After all layer
 buffers are flushed, 
 we begin to flush level buffers in $Q_o$. After each level flush, we
 rebalance all unbalanced base sets in this level, and if the level
 overflows we add the
 representative of this level to the end of $L_o$. Note that the
 representatives in $L_o$ are sorted on the minimum keys of the levels.

 After all overflowed level buffers are flushed, we enter the push
 stage and start to rebalance levels
in $L_o$  in a bottom-up fashion. In each step, we take out the first
level in $L_o$ (which is also the current lowest overflowed level) and
rebalance it. Suppose this level is level $j$ of
layer $X$.  If it is not the top level or the head layer we perform a level push;
otherwise we perform a layer push. Then we delete the representative
of this level from $L_o$.
 A level push may cause the level
buffer of level  $j+1$ to be overflowed, in which case we flush it and
rebalance the overflowed base sets. Then we check whether level $j+1$
overflows. If so, we insert the representative of level $j+1$ to
the head of  $L_o$ (unless it is already at the beginning of $L_o$) and
perform a level push on level
$j+1$. Otherwise we take out a new level in $L_o$ and continue the process.
When the top level of layer $N$ become unbalanced we simply perform a
global rebuild. 

After rebalancing all levels, we enter the pull stage and start to
process the delete signals. This is done as follows. 
We first process all delete signals in the head. If the head
becomes empty we perform a layer pull to get more keys into
the head. This may cause higher levels or layers to underflow, and we
keep performing level pulls and layer pulls until all levels and
layers are balanced. Consider a level pull or layer pull on level $j$
of layer $X$. After the level pull or layer pull the level buffer
$\Bcal_j$ may overflow. If so, we flush it and rebalance the base sets
when necessary. Note that this may cause the size of level $j$ to
grow, but it will not overflow, as we will show later, so that we do
not need push operations in the pull stage.   After all levels and
layers are balanced,  we process the delete signals in the head
again. We repeat the pull process until 
there are no delete signals left in the head and the head is non-empty.

\subsection{Correctness} 
It should be obvious that the flush and the push stage will always
succeed. The following two lemmas guarantee that the pull stage will
also succeed.
\begin{lemma}
\label{lem:supply}
When we perform a level pull on level $j$, there are
enough keys in level $j+1$ to rebalance level $j$; 
When we perform a layer pull on layer $X$, there are  enough keys
in level $0$ of layer $\Psi_X$ to rebalance level $l_X$.
\end{lemma}
\begin{proof}
Recall that a level pull on level $j$ transfers at most $4\cdot 8^j\Phi_X$ keys
from level $j+1$ to level $j$. 
Since we always perform pull operations in a
bottom-up fashion in the pull stage, and 
 all levels and layers are balanced before the pull stage, it follows that
level $j+1$ is always balanced when performing a pull operation on
level $j$. This implies that level $j+1$ has at least $2\cdot
8^{j+1}\Phi_X$ keys when performing a pull operation on level $j$, which
is sufficient to supply the level pull operation.

For a layer pull on layer $X$ other than the head, recall that the operation
transfers at most $4\cdot 8^{l_X}\Phi_X$ keys
from level $0$ of layer $\Psi_X$ to level $l_X$. By similar argument
we know level $0$ is balanced, so it has at least $2X$ keys. Following
Lemma~\ref{lem:top_level}, we have $2X\ge 8\cdot 8^{l_X}\Phi_X$, so
it suffices to supply the layer pull operation. The same argument also
works for a layer
pull on the head, since it acquires at most $cB$ keys from the upper
level, and the level contains at least $2cB$ keys. 
\end{proof}

\begin{lemma}
\label{lem:stage}
A level or a layer never overflows in the pull stage. 
\end{lemma}
\begin{proof}
Consider a level pull on level $j$ of layer $X$.  
Recall that since we move some keys from $\Bcal_{j+1}$ to $\Bcal_j$,
it is possible that $\Bcal_j$ overflows and  we need to perform a
level flush on level $j$.  We claim that after this level flush,  level $j$ is still balanced. 
For a proof, observe that level $j$ has size between
 $4\cdot 8^j\Phi_X$ and $(4\cdot
8^j-2)\Phi_X$ after the level pull, so it takes at least
$2\cdot 8^jB\log\frac{X}{B}\ge 2\cdot 8^jcB$ new updates before level
$j$ overflows. 
Since the level pull transfers at most $8^{j+1}B$ keys from
$\Bcal_{j+1}$, after the level pull, $\Bcal_j$ has less or
equal to $8^jB+8^{j+1}B=9\cdot 8^jB$ keys. Setting $c\ge 5$ allows
level $j$ to be still balanced after the level flush. This proves that 
a level never overflows in the pull stage.

Now consider a layer pull on layer $X$ other than the head. 
Recall that we move some keys
from the layer buffer $\Bcal_{\Psi_X}$ and level buffer $\Bcal_0$ to
$\Bcal_{l_X}$, it is possible that $\Bcal_{l_X}$ overflows and we need to
perform a level flush on the new level $l_X$. We claim that after this
level flush, level $l_X$ is still balanced.
For a proof, observe that after the layer pull, it takes at least $36\cdot 8^{l_X}B\log
\frac{X}{B}$ new updates before level $l_X$ overflows.
Since the layer pull transfers at most $X/2$ keys from the layer
buffer $\Bcal_{\Psi_X}$ and at most $8B$ keys from the  level
buffer $\Bcal_0$ to the level buffer of level $l_X$,   the level flush
operation flushes  at most $X/2+8B+8^{l_X}B$ keys to level $l_X$. 
By Lemma~\ref{lem:top_level}  we have 
\begin{eqnarray*}
36\cdot 8^{l_X}\Phi_X&=&20\cdot 8^{l_X}\Phi_X +16\cdot 8^{l_X}\Phi_X\\
&\ge& X/2+16\cdot 8^{l_X}B\log\frac{X}{B} \\
&\ge& X/2+8B+8^{l_X}B.
\end{eqnarray*} 
So level $l_X$ is still balanced after the level
flush. Finally, consider a layer pull on the head. Recall that it takes at
least $cB$ new update to the head before it overflows. Since the head
 acquires at most $cB/2$
keys from the layer buffer $\Bcal_{\Psi_{cB}}$, and at most $8B$ keys
from the level buffer $\Bcal_{0}$,  we can set $c > 16$ such that
$cB>cB/2+8B$, so the
head will remain balanced after the layer pull.  
This proves that a layer never overflows in the pull stage.
\end{proof}


\section{Analysis of Amortized I/O Complexity}
We analyze the amortized I/O cost for each operation during $N/8$
updates. We will show that the amortized I/O cost per update is
bounded by $O(\frac{1}{B}\sum_{i=0}S(B\log^{(i)}\frac{N}{B}))$, and
Theorem~\ref{thm:priority_queue} will follow.

\paragraph{Global rebuild.}
Recall that the I/O cost for a global rebuild is $O(NS(N)/B)$ I/Os. 
We claim that during
$N/8$ updates only a constant number of global rebuilds are needed, so
the amortized I/O cost per update is bounded by $O(S(N)/B)$. This can
be verified by the fact that a global
rebuild can only be triggered by $N/8$ new updates or that the level
$l_N$ becomes unbalanced,    and after a global rebuild
 it takes $\Omega(N)$ updates before level $l_N$ becomes
 unbalanced again.

\paragraph{Flush.}
We analyze the I/O cost of three different flush operations. For
memory flush, we sort a set of $B$
keys in the memory and merge them with a navigation list of size $O(\log^*
N)$. By Lemma~\ref{lem:flush}, the I/O cost is $O(\log ^* N)$. Therefore we  charge $O(\frac{\log^* N}{B})$ I/Os
for each of the updates in the memory buffer. Now consider a layer flush
at layer $X$. Let $\abs{\Bcal_X}$ denote the number of updates in the layer
buffer. By Invariant~\ref{inv:buffer_size} the layer flush operation is performed only if 
$\abs{\Bcal_X}\ge B\log\frac{X}{B}$.  
There are $O(\log\frac{X}{\Phi_X})$ level buffers, so by Lemma~\ref{lem:flush}, the I/O cost is
\begin{eqnarray*}
O\left(\frac{\abs{\Bcal_X}S(\abs{\Bcal_X})}{B}+\log\frac{X}{\Phi_X}\right)
&=&O\left(\frac{\abs{\Bcal_X}S(\abs{\Bcal_X})}{B}+\frac{\abs{\Bcal_X}}{B}\right)\\
&=&O\left(\frac{\abs{\Bcal_X}S(\abs{\Bcal_X})}{B}\right)\\
&=&O\left(\frac{\abs{\Bcal_X}S(N)}{B}\right).
\end{eqnarray*}
Thus, we can charge $O(S(N)/B)$ I/Os for each of
the $\Bcal$ updates in the layer buffer. 

Next, consider a level flush at
level $j$ in layer $X$. Let
$\abs{\Bcal_j}$ denote the number of updates in the buffer when we
perform the flush operation, and by Invariant~\ref{inv:buffer_size}
we have $\abs{\Bcal_j}\ge 8^jB$.
Recall that the size of the navigation
list is $\Theta(8^j)$, so by Lemma~\ref{lem:flush}, the I/O cost is
\begin{eqnarray*}
O\left(\frac{|\Bcal_j|S(|\Bcal_j|)}{B}+8^j \right)&=&O
\left(\frac{|\Bcal_j|S(|\Bcal_j|)}{B}+\frac{\abs{\Bcal_j}}{B}\right)
\\
&=& O\left(\frac{|\Bcal_j|S(|\Bcal_j|)}{B}\right)
=O\left(\frac{|\Bcal_j|S(N)}{B}\right). 
\end{eqnarray*}
Therefore we can charge $O(S(N)/B)$ I/Os for
each of the $\abs{\Bcal_j}$ updates in the level buffer. 

\paragraph{Rebalancing the base sets.}
Consider a rebalance operation for a base set $A$ at layer $X$.  When
$A$ overflows we sort  and divide it into equal segments. So the I/O cost
for a base set rebalance can be bounded by the sorting time of
$O(\abs{A}+\Phi_X)$ updates. Note that there are at least
$\Omega(\abs{A})$ updates to $A$ since the last rebalance
operation on it, and by Invariant~\ref{inv:level_size} 
we have $\abs{A}\ge 2\Phi_X$. Thus, the
amortized I/O cost per update is $O(S(N)/B) $.  
  
\paragraph{Rebalancing the levels.} We first consider a level push operation
  on level $j$ of layer $X$. The operation cuts the base navigation
  list of level $j$, takes the first half
  to form a new level $j$, and attaches the rest to  level $j+1$. The
  I/O cost for this cut-attach procedure is 
 $\Theta(8^j/B+1)$,  since the navigation list is sorted
  and stored consecutively on disk.
 Then the operation sorts and redistributes the level buffer
 $\Bcal_j$. Recall that we always flush the level
 buffer before rebalancing the level, so we have $\abs{\Bcal} \le 8^jB$
 when the level push is performed.  The I/O cost  for sorting and
 redistributing $\Bcal_j$ is bounded by
  $O(8^jB\cdot S(8^jB)/B)=O(8^jS(X))$. Note that after a level push,
  it takes at least $\Theta(8^j\Phi_X)$ new updates to
  level $j$ before it overflows again. So during $N/8$ updates 
at most $O(N/ (8^j\Phi_X))$ level push operations are performed on level
  $j$. It follows that 
  the I/O cost of all level push operations on level
  $j$ is bounded by 
$$O\left(8^jS(X) \cdot \frac{N}{8^j\Phi_X}\right)=
O\left(\frac{NS(X)}{B\log\frac{X}{B}}\right).$$ 
We charge $O(S(X)/B\log\frac{X}{B})$ for
  each update and for each level in layer $X$.
Since there are $O(\log\frac{X}{B})$ levels in layer $X$, 
 we charge  $O(S(X)/B)$ I/Os for
  each update in layer $X$. Summing up all layers, the
  amortized I/O cost for each update is
  $O(\frac{1}{B}\sum_{i=0}S(B\log^{(i)}\frac{N}{B}))$.
A similar argument shows that the amortized I/O cost for the
level pulls is the same,  except that  the I/O cost
is amortized only on the delete signals.  

\paragraph{Rebalancing the layers.} 
Consider a layer 
push operation on layer $X$.  It takes the
keys in level $l_X$ of layer $X$ and level $0$ of layer $\Psi_X$,
sorts them, and rebuilds both levels. Since both levels
have size $O(X)$, the I/O cost is $O(XS(X)/B)$. We also note that after a
layer push operation, it takes at least $\Theta(X)$ updates to level
$l_X$  before it goes unbalanced
again. That means at most $O(N/X)$ layer rebalance operations are
needed. So the I/O cost for the layer rebalances of
layer $X$ during the $N/8$ updates is $O(NS(X)/B)$. We can charge
$O(S(X)/B)$ I/Os for each update and each layer,  and summing up all
layers,  it is amortized 
$\frac{1}{B}\sum_{i=0} S(B\log^{(i)}\frac{N}{B})$ I/Os for each
update. Similar argument shows that the amortized I/O cost for
layer pulls is the same,  except that the total I/O cost
is amortized only on the delete signals.  

\paragraph{Scheduling the operations.} Note that in the schedule we
need to pay some extra I/Os for
maintaining the queue $Q_o$ and doubly linked list $L_o$.
We observe that an update to $Q_o$ or
$L_o$ would trigger a flush or rebalance operation later, and the cost of a flush or
rebalance operation  is at least $1$ I/O. So $Q_o$ and $L_o$ can be
maintained without increasing the asymptotic I/O cost.

\bibliographystyle{abbrv} 
\bibliography{/csproject/yike/paper,/csproject/yike/io,/csproject/yike/geom,/csproject/yike/newgeom}

\end{document}